\documentclass[conference]{IEEEtran}
\usepackage{amssymb}
\usepackage{amsmath}
\usepackage{amsfonts}
\usepackage{graphicx}
\usepackage{algorithm}
\usepackage{algorithmic}
\usepackage{cite}
\usepackage{epstopdf}
\usepackage{color}
\usepackage{enumerate}
\usepackage{stfloats}
\newtheorem{theorem}{\textbf{Theorem}}

\newtheorem{proposition}[theorem]{\textbf{Proposition}}
\newtheorem{remark}{\textbf{Remark}}

\newtheorem{proof}{\textbf{Proof}}
\begin{document}
\title{Delay Outage Probability of Multi-relay Selection for Mobile Relay Edge Computing Systems}
\author{Jie Liang, Zhiyong Chen, Cheng Li, Bin Xia, and Ning Liu\\
Cooperative Medianet Innovation Center, Shanghai Jiao Tong University,  Shanghai,  P. R. China\\
Email: \{liangjie, zhiyongchen, lichengg, bxia, ningliu\}@sjtu.edu.cn }
\maketitle
\begin{abstract}
In this paper, we deal with the problem of relay selection in mobile edge computing networks, where a source node transmits a computation-intensive task to a destination via the aid of multiple relay nodes. It differs from the traditional relay nodes in the way that each relay node is equipped with an edge computing server, thus each relay node can execute the received task and forwards the computed result to the destination. Accordingly, we define a delay outage probability to evaluate the impact of the relay computing ability on the link communications, and then propose a latency-best relay selection (LBRS) scheme that not only consider the communication capability, but also consider the computing ability. The performance of the proposed relay selection scheme with the traditional communication-only relay selection (CORS) and computing-only selection (CPORS) schemes in terms of the delay outage probability and the diversity order is analyzed, compared with other relay selection schemes. We show that the proposed LBRS scheme reduces to the traditional CORS scheme under the high signal-to-noise ratio (SNR) region. We further reveal that the diversity orders of both the proposed LBRS and CORS schemes are dependent on the computing ability of relay nodes.
\end{abstract}
%\begin{keywords}
%Delay outage probability, diversity order, edge computing, relay selection.
%\end{keywords}

\section{Introduction}
In the era of mobile internet, various mobile computing services are emerging, e.g., vehicular navigation, image recognition and augmented reality (AR). Such computation-intensive services are delay sensitive and have more and more demands on the computing ability of the network, imposing a real challenge in the resource allocation of both communication and computing resources \cite{primary,peng}.

Motivated by the above analysis, we discuss a novel issue of relay selection in cooperative systems due to the introduction of computing in this paper. Considering a mobile relay edge systems, where a source $S$ transmits a computation-intensive task to a destination $D$ with the help of $N$ computing-enabled relay nodes, shown in Fig \ref{fig1}. Each relay node can execute the received task from $S$, and then forwards the computed results to $D$. One example of this model is mobile AR delivery in the vehicle network, where the road information, e.g., the road congestion waiting for reasons and expected waiting time collected by the front vehicle, can be computed and transmitted to the following vehicle by the relay node and displayed on the following vehicle in AR mode.
%\begin{figure}
%\begin{center}
%\includegraphics[width=0.40]{systemodel.eps}
%\end{center}
%\caption{Mobile Relay Edge Computing Systems}
%
%\end{figure}
\begin{figure}[t]
\begin{center}
 \includegraphics[width=2.8in]{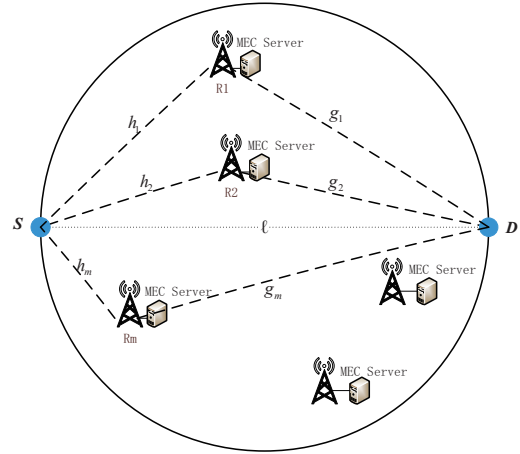}
\end{center}
 \caption{Mobile Relay Edge Computing Systems.}
\label{fig1}
\end{figure}
This mobile edge computing (MEC) architecture is different from the classical MEC model\cite{survey}. In the classical MEC, each user offloads its own computation task to the MEC server in the base station (BS), and the output of the computation task is then transmitted to the user by the BS. Most of these works about this traditional model consider the cost of offloading data in the uplink while ignoring the cost of the transmission in the downlink\cite{mecjsac,twcyou}. It is worth noting that a basic three-node MEC system, consisting of a user node, a helper node, and an AP node attached with a MEC server, is proposed in \cite{meciot}. The main contribution of \cite{meciot} is to optimize the task partition of the user to minimize the total energy consumption at both the user and the helper.

One of the key challenges in cooperative systems is relay selection. A variety of relay selection schemes for amplify-and-forward (AF) or decode-and-forward (DF) have been proposed \cite{Krikidis2015} to cooperatively forward data to the destination. Such communication cooperative relaying schemes only consider the communication resource. However, for our considered system, the relay node not only needs to forward the data but also needs to compute the task. As a result, the available computing ability of the relay node plays an important role in the relay selection policy. In this paper, we propose a latency-best relay selection (LBRS) scheme for the mobile relay edge computing system that not only considers the communication capability but also considers the computing ability.

Furthermore, the traditional outage probability is usually used to analyze the PHY layer behavior of the cooperative system, but cannot qualify the impact of the computing ability on the relay selection. To address this problem, we define a delay outage probability to evaluate the impact of the relay computing ability on the link communications. Accordingly, we derive the expressions of outage probability for the proposed LBRS scheme as well as conventional communication-only relay selection (CORS) and computing-only relay selection (CPORS) schemes. We show that the proposed LBRS scheme reduces to the conventional CORS scheme under the high signal-to-noise ratio (SNR) region. Additionally, we characterize the diversity orders of both the LBRS and CORS schemes as well as the CPORS scheme. s. It is demonstrated that the LBRS and CORS scheme not only have the same diversity order, but both are not always can achieve the full diversity order, depending on the computing ability of relay nodes. Finally, simulation results are done to validate the accuracy of our theoretical analysis, and show that the proposed LBRS scheme has the best performance compared with the CORS and CPORS schemes.

\section{System Model}

As shown in Fig \ref{fig1}, we assume there is no direct link between $S$ and $D$ and each node is equipped with one antenna. Let $h_i$ be the channel fading coefficient from $S$ to the $i$-th relay nodes $R_{i}$, and $g_i$ be the channel fading coefficient from $R_{i}$ to $D$, where we have $h_i \thicksim \mathcal{CN}(0,1)$ and $g_i \thicksim \mathcal{CN}(0,1)$ for $i=1,...,N$. Besides, let $d_{si}$ and $d_{di}$ be the distance between $S$ and $R_i$ and between $R_i$ and $D$, respectively. Thus we consider the path-loss channel coefficient is $1/(1+d_{si}^\alpha)$ and $1/(1+d_{di}^\alpha)$ for the $S$-$R_i$ link and $R_i$-$D$ link \cite{Krikidis2015}, respectively. Here, $\alpha > 2$ denotes the path-loss exponent.

The source $S$ has a computing task characterize by tuple $(L,K,\rho)$, where $L$ is the input size of the task (in \emph{bits}), $K$ is the number of required CPU cycles per bit (in \emph{cycles/bit}), and the output data size of the task divided by the input size $L$ is defined as the data computing ratio $\rho\geq 0$. Each relay node $R_i$ has a computing server, which can run at a constant CPU-cycle frequency $f_i$ (in \emph{cycles/s}).

Firstly, the source $S$ transmit the task to the relay nodes, and the received signal at the $i$-th relay node $R_i$ is
\begin{equation}
\setcounter{equation}{1}
y_i=\sqrt{P_s}\frac{h_i}{\sqrt{1+d_{si}^\alpha}}x_1+n_i.
\end{equation}
where $P_s$ is the transmit power of $S$, $x_1$ denotes the input data signal of the task with normalized power and $n_i$ is additive white Gaussian noise (AWGN) with variance $\sigma^2$. If the $i$-the relay node $R_i$ is active, the achievable transmission rate is
\begin{align}
\nu_{i}^1 =W \log_2(1+\gamma_{i}^1) = W \log_2 (1+ \frac{P_s\mid h_i \mid ^2}{(1+d_{si}^\alpha)\sigma^2 }).
\end{align}
Here, let $W$ be the bandwidth of each link in this paper. Therefore, the transmission time can be given by
\begin{align}
t_{i}^1 = \frac{L}{\nu_{i}^1} = \frac{L}{W \log_2 (1+ \frac{P_s\mid h_i \mid ^2}{(1+d_{si}^\alpha)\sigma^2})}.
\end{align}

In this paper, each relay node operates in the DF mode. The $i$-th relay node $R_i$ decodes $x_1$ from $y_i$ and then compute it by using edge computing server. The computation latency is then given by $t_{i}^c = \frac{LK}{f_i}$.

The output data signal $x_2$ is then transmitted to the destination $D$, and the received signal at $D$ is
\begin{equation}
y_d=\sqrt{P_{r}}\frac{g_i}{\sqrt{1+d_{ri}^\alpha}}x_2+n_d,
\end{equation}
where each relay node has the same transmit power $P_{r}$ and $n_d$ is additive white Gaussian noise (AWGN) with variance $\sigma^2$. Similarly, the achievable transmission rate and the transmission time are given by
\begin{align}
&\nu_{i}^2 =W \log_2(1+\gamma_{i}^2) = W \log_2 (1+ \frac{P_{r}\mid g_i \mid ^2}{(1+d_{ri}^\alpha)\sigma^2}).\\
&t_{i}^2 = \frac{\rho L}{\nu_{i}^2} = \frac{\rho L}{W \log_2 (1+ \frac{P_{ri}\mid g_i \mid ^2}{(1+d_{ri}^\alpha)\sigma^2})}.
\end{align}
respectively. As a result, the source-to-destination delay for completing the service with the aid of $R_i$ is $t_i = t_{i}^1+t_{i}^c+t_{i}^2$.
\begin{figure*}[b]
\vspace{2mm}
\hrulefill
\setcounter{equation}{16}
\begin{align}
&f_z(z)=\int_{0}^{\infty}f_x(x)f_y(z-x)dx \nonumber \\
&=\underbrace{\frac{(\ln^22) \rho (1+d_{si}^\alpha)\sigma^2(1+d_{ri}^\alpha)\sigma^2}{P_sP_{r}}e^{\frac{(1+d_{si}^\alpha)\sigma^2}{P_s}+\frac{(1+d_{ri}^\alpha)\sigma^2}{P_{r}}}}_{\triangleq O_i}\underbrace{\int^{z}_{0} e^{-(\frac{(1+d_{si}^\alpha)\sigma^2}{P_s}2^{\frac{\rho}{x}}+\frac{(1+d_{ri}^\alpha)\sigma^2}{P_{r}}2^{\frac{1}{z-x}})}\frac{\rho}{x^2(z-x)^2}2^{\frac{\rho}{x}+\frac{1}{z-x}}dx}_{\triangleq Q_i}.\label{equ:Zpdf}
\end{align}
\end{figure*}

\section{Relay Selection Scheme and Delay Outage Performance}

In this section, we present three relay selection schemes and discuss the delay outage probability performance, by taking into consideration the computing ability at the relay node. The delay outage probability is defined as the probability that the source-to-destination latency exceeds the maximum delay-bound $D_{\max}$, which can be express as $\Pr\{Delay\geq D_{\max}\}$. Note that we also assume in this paper that the source and destination have perfect channel state information (CSI) knowledge of all links.

\subsection{Communication-only relay selection}

The communication-only relay selection (CORS) scheme only considers the communication rate and selects the relay node whose has the maximum of the transmission rate, known as the best-relay selection scheme in cooperative communications systems. The collection $\theta=\{i|i=1,\cdots,N\}$ is defined, and the relay node selected by the CORS policy is
\begin{align}\label{equ:bix1}
\setcounter{equation}{6}
i^* = \arg_{i\in\theta}\max\limits_{i=1...N}\{\min(\nu_{i}^1,\nu_{i}^2)\}.
\end{align}

The delay outage probability of the CORS scheme can be described as\footnote{Because the relay node needs to compute the task, it does not need the same the transmission rate for the source-relay link and the relay-destination link, even if the relay node operates in the DF mode.}.
\begin{align}
P_{out}^{CORS}&=\Pr\{\frac{L}{\nu_{i^*}^1}+\frac{LK}{f_{i^*}}+\frac{\rho L}{\nu_{i^*}^2}>D_{\max}\}\nonumber\\
&\leq \Pr\{\frac{L+\rho L}{\nu}> D_{\max}-\frac{LK}{f_{i^*}}\}\label{equ:BC define},
\end{align}
where we use $\nu =\min(\nu_{i^*}^1,\nu_{i^*}^2)$, and it is easy to see that if $\nu_{i^*}^1=\nu_{i^*}^2$ or $P_{s}\to \infty$ or $P_{r}\to \infty$, the equality in (\ref{equ:BC define}) clearly holds. Based on (\ref{equ:bix1}), the exact expression of delay outage probability of the CORS scheme is very difficult to obtain in (\ref{equ:BC define}). Thus we derive the upper bound of the delay outage probability for the CORS scheme based on (\ref{equ:BC define}).

Define $\varphi_i=D_{\max}-\frac{LK}{f_i}$ and $\phi=\{i|\varphi_i>0,i=1\cdots,N\}$. It is easy to see that when $\varphi_{i^*}<0$, the outage probability is 1. As a result, if $i^*\notin\phi$, we have $P_{out}^{CORS}=1$. If $i^*\in\phi$, we obtain

\begin{align}
&\Pr\{\nu <\frac{L+\rho L}{\varphi_{i^*}}\}=\Pr\{\max\limits_{i\in\theta,\varphi_i>0}\{\min(\nu_{i}^1,\nu_{i}^2)\}<\frac{L+\rho L}{\varphi_i}\} \nonumber\\
%&=\prod\limits_{i=1,i\in\phi}^N \Pr\{\min(\nu_{i}^1,\nu_{i}^2)<\frac{L+\rho L}{\varphi_i}\} \nonumber\\
%&=\prod\limits_{i=1,i\in\phi}^N[1-\Pr\{\min(\nu_{i}^1,\nu_{i}^2)>\frac{L+\rho L}{\varphi_i}\}] \nonumber\\
&=\prod\limits_{i=1,i\in\phi}^N[1-\Pr\{\nu_{i}^1>\frac{L+\rho L}{\varphi_i}\}\Pr\{\nu_{i}^2>\frac{L+\rho L}{\varphi_i}\}]\label{equ:rate}.
\end{align}

For $h_i \thicksim \mathcal{CN}(0,1)$, $g_i \thicksim \mathcal{CN}(0,1)$, $\mid h_i\mid ^2$ and $\mid g_i \mid ^2$ are the exponential distribution with parameter $\lambda = 1$, respectively. We have
\begin{equation}
\Pr\{\nu_{i}^1>\frac{L+\rho L}{\varphi_i}\}=\exp\{-\frac{(1+d^\alpha_{si})\sigma^2 (2^{\frac{L+\rho L}{W\varphi_i}}-1)}{P_{s}}\}.\label{equ:erate}
\end{equation}
Substituting (\ref{equ:erate}) into (\ref{equ:rate}), the upper bound is
\begin{align}
P_{out}^{CORS}\leq &\prod\limits_{i=1,i\in\phi}^{N} \left\{1-\exp\{-\frac{(1+d^\alpha_{si})\sigma^2 (2^{\frac{L+\rho L}{W\varphi_i}}-1)}{P_{s}}\}\right. \nonumber\\
&\times \left.\exp\{-\frac{(1+d^\alpha_{ri})\sigma^2 (2^{\frac{L+\rho L}{W\varphi_i}}-1)}{P_{r}}\}\right\}\label{equ:maxmin_r}.
\end{align}
\begin{remark}
Condsidering the special case with $f_i\to [\frac{LK}{D_{\max}}]^+$, $\forall i \in\theta$, i.e., $\varphi_i \to 0$, the outage probability of CORS scheme is given by $\lim_{f_i\to [\frac{LK}{D_{\max}}]^+}P_{out}^{CORS} =1$. If $P_s\to\infty$ as a special condition, $\nu_i^1>>\nu_i^2$, $\forall i \in \theta$ can be derived. The uplink transmission time can be ignored and the outage probability is
\begin{align}
\lim_{P_s\to\infty}P_{out}^{CORS} \!\!=\!\!\!\prod\limits_{i=1,i\in\phi}^N\!\!\![1-\exp\{-\frac{(1+d^\alpha_{ri})\sigma^2 (2^{\frac{\rho L}{W\varphi_i}}-1)}{P_{r}}\}],
\end{align}
for $i^*\in\phi$. Similarly, if $P_{r}\to \infty$, the outage probability can be obtained as following
\begin{align}
\lim_{P_{r}\to\infty}P_{out}^{CORS}\!\!=\!\!\prod\limits_{i=1,i\in\phi}^N[1-\exp\{-\frac{(1+d^\alpha_{si})\sigma^2 (2^{\frac{L}{W\varphi_i}}-1)}{P_{s}}\}],
\end{align}
\end{remark}
for $i^*\in\phi$.

For the outage probability performance, we can see that the CORS scheme (traditional best-relay selection) does not consider the computing ability in the relay node, so that suffers from an outage pulse, which depends on the computing ability $f_{i^*}$ of the selected relay node.

\vspace{-4mm}
\subsection{Computing-only relay selection}

In contrary to the CORS policy, the computing-only relay selection (CPORS) only consider the computing ability of the relay nodes and selects the relay node with the maximum of computing ability. Specifically, the selection criterion of the CPORS policy is given by $i^* = \arg_{i\in\theta}\max_{i=1...N}\{f_i\}$. Let $f_{i^*} = \max\limits_{i=1}^{N}\{f_i\}$.  The delay outage probability is given by
\begin{align}
&P_{out}^{CPORS}=\Pr\{Y+X\geq \frac{W\varphi_{i^*}}{L}\}, \label{equ:pr1}
\end{align}
where $Y=\frac{1}{\log_2 (1+ \frac{P_s\mid h_i \mid ^2}{(1+d_{si}^\alpha)\sigma^2})}$ and $X=\frac{\rho }{ \log_2 (1+ \frac{P_{r}\mid g_i \mid ^2}{(1+d_{ri}^\alpha)\sigma^2})}$, then we have cumulative distribution function (cdf) of $Y$ and $X$
\begin{align}
F(Y)&=\Pr\{\frac{1}{\log_2 (1+ \frac{P_s\mid h_i \mid ^2}{(1+d_{si}^\alpha)\sigma^2})}\leq y\}=e^{-\frac{(1+d_{si}^\alpha)\sigma^2}{P_s}(2^{\frac{1}{y}}-1)},\nonumber\\
F(X)&=e^{-\frac{(1+d_{ri}^\alpha)\sigma^2}{P_{r}}(2^{\frac{1}{x}}-1)}. \label{pdf_xy}
\end{align}
respectively. Based on (\ref{pdf_xy}), the probability density function (pdf) of $Y$ and $X$ can be written as
\begin{align}
f_y(y)&=\frac{dF(Y)}{dy}=\frac{(1+d_{si}^\alpha)\sigma^2\ln 2}{P_s y^2}2^{\frac{1}{y}}e^{-\frac{(1+d_{si}^\alpha)\sigma^2}{P_s}(2^{\frac{1}{y}}-1)}.\nonumber\\
f_x(x)&=\frac{(1+d_{ri}^\alpha)\sigma^2 \rho \ln 2}{P_{r} x^2}2^{\frac{\rho}{x}}e^{-\frac{(1+d_{ri}^\alpha)\sigma^2}{P_{r}}(2^{\frac{\rho}{x}}-1)}.
\end{align}

Let $Z=X+Y$. Based on the convolution formula, the pdf of $Z$ is given by (\ref{equ:Zpdf}), as shown at the bottom of the page. Substituting (\ref{equ:Zpdf}) into (\ref{equ:pr1}), the outage probability can be expressed as
\begin{align}
\setcounter{equation}{17}
P_{out}^{CPORS}\!\!=\!\Pr\{X+Y\geq\frac{W\varphi_{i^*}}{L}\}\!\!=\!\!\left\{\begin{array}{ll}
O_{i^*}\int^\infty_{\frac{W\varphi_{i^*}}{L}}Q_{i^*} \!\!&\!\!i^*\in\phi\\
1\!\!&\!\!i^*\notin\phi.
\end{array}\right.\label{equ:maxf}
\end{align}

\begin{remark}
For the special case with $P_s \to \infty$, the outage probability of the CPORS scheme is given by
\begin{align}
&\lim_{P_s\to\infty}P_{out}^{CPORS} = Pr\{\frac{LK}{f_{i^*}}+\frac{\rho L}{\nu_{{i^*}}^2} \geq D_{\max}\}\nonumber\\
&=1-\exp\{-\frac{(1+d_{r{i^*}}^\alpha)\sigma^2}{P_{r}}(2^{\frac{\rho L}{\varphi_{i^*} W}}-1)\}.
\end{align}
Similarly, when $P_{r}\to\infty$, the outage probability is given by
\begin{align}
&\lim_{P_{r}\to\infty}P_{out}^{CPORS}= Pr\{\frac{L}{\nu_{{i^*}}^1}+\frac{LK}{f_{i^*}} \geq D_{\max}\}\nonumber\\
&=1-\exp\{-\frac{(1+d_{s{i^*}}^\alpha)\sigma^2}{P_{s}}(2^{\frac{L}{\varphi_{i^*} W}}-1)\}.
\end{align}
\end{remark}
\vspace{-6mm}
\subsection{Latency-best relay selection}
\vspace{-1mm}
The minimal total latency (LBRS) scheme means we always select the relay node whose the corresponding latency is the minimal one among $N$ relays. According to the definition of the delay outage probability, it is obvious that LBRS is optimal. The relay node selected by the LBRS policy can be expressed as $i^* = \arg_{i\in\theta}\min_{i=1...N}\{t_i\}$.The outage probability can be described as
\begin{align}
&\Pr\{\min\limits_{i=1,i\in \theta}^N {t_i}\geq D_{max}\}=\prod\limits_{i=1, i\in\phi}^{N} Pr\{t_{i}^1+t_{i}^c+t_{i}^2\geq D_{max}\}\nonumber.
\end{align}
Based on (\ref{equ:maxf}), the outage probability of the LBRS policy is
\begin{align}
\prod\limits_{i=1,i\in\phi}^{N} Pr\{t_{i}^1+t_{i}^2\geq \varphi_i\}=\prod\limits_{i=1,i\in\phi}^{N} O_i\int^\infty_{\frac{W\varphi_i}{L}}Q_i.
\end{align}
\begin{remark}
 For the special case with $P_s\to \infty$, the outage probability of the LBRS policy is given by
 \begin{align}
&\lim_{P_s\to\infty}P_{out}^{LBRS} = \prod\limits_{i=1,i\in\phi}^{N} Pr\{t_{i}^c+t_{i}^2\geq D_{max}\}\nonumber\\
&=\prod\limits_{i=1,i\in\phi}^{N}[1-\exp\{-\frac{(1+d_{ri}^\alpha)\sigma^2}{P_{r}}(2^{\frac{\rho L}{\varphi_i W}}-1)\}],\nonumber\\
&=\lim_{P_s\to\infty}P_{out}^{CORS}.
 \end{align}
 Similarly, if $P_{r}\to \infty$, the outage probability is
 \begin{align}
&\lim_{P_{r}\to\infty}P_{out}^{LBRS}= \prod\limits_{i=1,i\in\phi}^{N} Pr\{t_{i}^1+t_{i}^c\geq D_{max}\}\nonumber\\
&=\prod\limits_{i=1,i\in\phi}^{N}[1-\exp\{-\frac{(1+d_{si}^\alpha)\sigma^2}{P_{s}}(2^{\frac{L}{\varphi_i W}}-1)\}],\nonumber\\
&=\lim_{P_{r}\to\infty}P_{out}^{CORS}.
 \end{align}
\end{remark}

Remark 3 shows that for high SNRs, the outage probability of the LBRS policy is equal to that of the CORS policy. The intuition behind
this result is that in the high SNR regime of one link, the optimal relay selection (LBRS) is independent of the computing ability of the relay node (CORS).

\section{Diversity order analysis}
In this section, we analyze the diversity order of the three schemes. Following the \cite{chen2012}, the diversity order is defined as $d=-\lim_{\gamma\to \infty}\frac{\log P_{out}(\gamma)}{\log \gamma}$, where $\gamma$ is an SNR and $P_{out}(\gamma)$ is the delay outage probability function of $\gamma$. In order to computing the diversity order, the distance between relay and source or destination can be infinitely closed to 0 for $\gamma \to \infty$ \cite{chen2012}. Supposed $P_s=P_{r}$, $\beta_i=\frac{(1+d^\alpha_{ri})}{1+d^\alpha_{si}}$,$\gamma = \frac{P_s}{(1+d_{si}^\alpha)\sigma^2}$, $\gamma$ can be seen as a constant independent of distance for $\gamma \to \infty$.
\begin{theorem}
The diversity order of the CORS scheme for the mobile relay edge computing network is
\begin{align}
d_{CORS} = |\phi| \leq N,
\end{align}
where $|\phi|$ is the cardinality of the relay node set $\phi$.
\end{theorem}
\begin{proof}
For the CORS scheme, (\ref{equ:maxmin_r}) can be rewritten as
 \begin{align}
P_{out}^{CORS}\!\!\!\leq \!\!\!\!&\!\!\prod\limits_{i=1,i\in\phi}^{N}\!\!\! \left\{\!1\!-\!\exp\{\frac{1-2^{\frac{L+\rho L}{W\varphi_i}}}{\gamma}\}\right.\left.\exp\{\frac{1-2^{\frac{L+\rho L}{W\varphi_i}}}{\gamma}\beta_i\}\right\}. \label{diversity}
\end{align}
Using the Taylor expansion for $\gamma\to \infty$
\begin{align}
\lim_{\gamma\to\infty} \exp\{-\frac{2^{\frac{L+\rho L}{W\varphi_i}}-1}{\gamma}\}=1-\frac{2^{\frac{L+\rho L}{W\varphi_i}}-1}{\gamma} \label{Taylor1},
\end{align}
\begin{align}
\lim_{\gamma\to\infty} \exp\{-\frac{(2^{\frac{L+\rho L}{W\varphi_i}}-1)}{\gamma}\frac{(1+d^\alpha_{ri})}{1+d^\alpha_{si}}\} = 1-\frac{(2^{\frac{L+\rho L}{W\varphi_i}}-1)}{\gamma}\beta_i \nonumber.
\end{align}

Therefore, we have
\begin{align}
P_{out}^{CORS}= &\frac{1}{\gamma^{|\phi|}}\prod\limits_{i=1,i\in\phi}^{N}\left\{(2^{\frac{L+\rho L}{W\varphi_i}}-1)\frac{(2+d^\alpha_{si}+d^\alpha_{ri})}{1+d^\alpha_{si}}\right\}, \nonumber
\end{align}
for $\gamma\to\infty$. Thus, we have the theorem.
\end{proof}

Theorem 1 implies that the diversity order of the CORS scheme is less than or equal to N, and it is dependent on the computing ability of relay nodes.

From \cite{DenizG2008} the diversity order of multi-hop relay channel can be characeterized by $d= \min\{d_{S,R},d_{R,D}\}$.

If $\varphi_{i^*}\leq0$, the diversity order of the CPORS scheme is 0, and when $\varphi_{i^*}>0$, the outage probability of the source-relay link and the relay-destination link are given by
\begin{align}
&\Pr\{t_{i^*}^1+t_{i^*}^c\}=(2^{\frac{L}{\varphi_{i^*} W}}-1)\frac{1}{\gamma},\\
&\Pr\{t_{i^*}^c+t_{i^*}^2\}=(2^{\frac{L}{\varphi_{i^*} W}}-1)\beta_i\frac{1}{\gamma},
\end{align}
respectively, for $\gamma\to\infty$. We thus have $d_{CPORS}=1$, when $\varphi_{i^*}>0$. As a result, we have the following proposition for the CPORS policy.
\begin{proposition}
The diversity order of the CPORS scheme for the mobile relay edge computing network is
\begin{align}
d_{CPORS} =\left\{\begin{array}{ll}
1, &\text{if}~ \varphi_{i^*}>0\\
0,&\text{if}~  \varphi_{i^*}\leq0.
\end{array}\right.
\end{align}
\end{proposition}

Obviously, when the computing ability of the relay nodes are different, there is only one option: the maximum computing ability relay node and the computing ability determines the diversity order.

In the LBRS scheme, the outage probability of the source-relay link and the outage probability of the relay-destination link are given by
\begin{align}
&\prod\limits_{i=1,i\in\phi}^{N} Pr\{t_{i}^1+t_{i}^c\geq D_{max}\}=\frac{1}{\gamma^{|\phi|}}\prod\limits_{i=1,i\in\phi}^{N}2^{\frac{L}{\varphi_i W}}-1,\nonumber\\
&\prod\limits_{i=1,i\in\phi}^{N} Pr\{t_{i}^c+t_{i}^2\geq D_{max}\}=\frac{1}{\gamma^{|\phi|}}\prod\limits_{i=1,i\in\phi}^{N} (2^{\frac{L+\rho L}{W\varphi_i}}-1)\beta_i,\nonumber
\end{align}
respectively, for $\gamma \to \infty$. Accordingly, the diversity order of the LBRS scheme is $d_{LBRS}=|\phi|\leq N$, which is same with $d_{CORS}$. This suggests that the optimal relay selection scheme (CORS) also can not always achieve the full diversity gain $N$ for traditional relay nodes without computing.
\begin{figure}

\begin{center}
\includegraphics[width=3.2in]{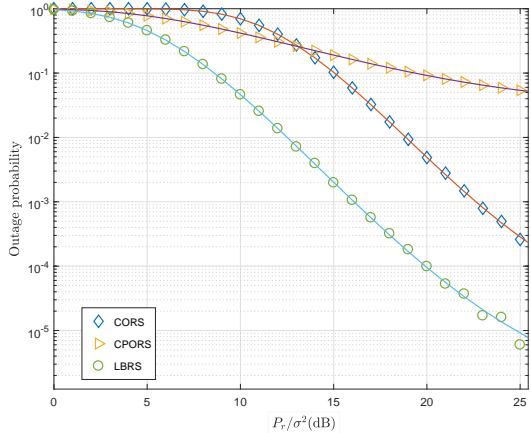}
\caption{Outage probability versus $P_{r}/\sigma^2$. Here, the solid lines represent the theoretical results.}
\label{fig2}
\end{center}

\end{figure}

\begin{figure}
\begin{center}
\includegraphics[width=3.2in]{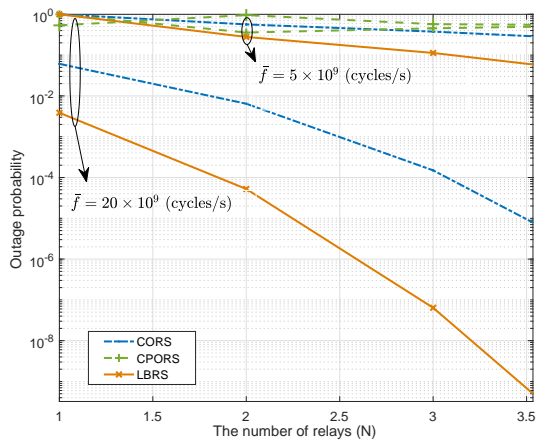}
\caption{Outage probability versus the number of relays. Here, we use $P_{r}/\sigma^2$=20 dB.}
\label{fig3}
\end{center}
\vspace{-3mm}
\end{figure}

\section{Simulation Results}
In this section, we present simulation results to evaluate the performance of the three schemes. Unless specified, the following parameters are used throughout this section: $\ell=1$, $L = 50\times 10^6$ bits, $K=10$ cycles/bits, $\rho=0.5$, $P_s/\sigma^2=25$ dB, $W=100\times 10^6$ Hz, $N=4$, $D_{\max}=0.2$ s and $\alpha =3$. Without loss of generality, we set the CPU speed $f$ of relays is the uniform distribution between $f \sim (2\times2.5\times10^{9},12\times2.5\times 10^{9})$ (cycles/s) \cite{VR}.

Fig.\ref{fig2} shows the outage probability of the proposed relay schemes versus $P_{r}/\sigma^2$. We also plot simulation results to validate the accuracy of our theoretical analysis in Fig.\ref{fig2}. We can see that the theoretical results (solid lines) are very close to simulated results for different $P_{r}/\sigma^2$ values. It can be also seen from the figure that the LBRS scheme has the best performance in all $P_{r}/\sigma^2$ values, when the outage probability using the CORS scheme is slightly higher than that of the CPORS scheme in the low $P_{r}/\sigma^2$ region ($\leq 13$ dB) but then is significant lower in the high $P_{r}/\sigma^2$ region. This is because that the CPORS scheme takes into account the computing power of the relay to compensate for the delay loss caused by a certain low transmission rate.

Fig.\ref{fig3} depicts the outage probability of the three schemes versus the number of relays with different computing ability. Taking the mean of $5\times 10^9$ (cycles/s) and $20\times 10^9$ (cycles/s) as examples relays with high average CPU-cycle frequency, the overall outage probability is less than that of relays with a low average CPU-cycle frequency. As the number of relays increases, the outage probability decreases, and the relays with high average CPU-cycle frequency has a faster drop rate.

\section{Conclusion}
In this paper, we have proposed a new relay selection scheme LBRS for mobile relay edge computing system, consisting of a source, a destination, and multiple relays which have computing servers to cooperative computing. In order to analyze the cooperative communications with computing, a delay outage probability is defined and then we have investigated the performance for the proposed LBRS scheme with two traditional relay selection schemes in terms of the delay outage probability and the diversity order. The key results of this work can be summarized as follows: $i$) The computing ability of relay nodes plays an important role in the relay selection policy; $ii$) The proposed LBRS scheme has a good performance in terms of the delay outage probability and the diversity order; $iii$) The proposed LBRS scheme reduces to the traditional CORS scheme for the high SNR region.

\bibliographystyle{IEEEtran}{}
\bibliography{VOD_paper}

%-----------------------------------------------------
\end{document}